\documentclass[10pt,journal,letterpaper]{IEEEtran}

\usepackage{amsmath,amssymb,amsfonts}
\usepackage{mathtools}
\usepackage{graphicx}
\usepackage{booktabs}
\usepackage{microtype}
\usepackage{xspace}
\usepackage{url}
\usepackage{hyperref}
\usepackage{algorithm}
\usepackage[noend]{algpseudocode}

\usepackage{textcomp}
\usepackage{stfloats}
\usepackage{url}
\usepackage{verbatim}
\usepackage{graphicx}
\usepackage{cite}
\usepackage{amssymb}
\usepackage{ifthen}
\usepackage[utf8]{inputenc}
\usepackage{latexsym, dsfont}
\usepackage{amsmath,amsthm}
\usepackage{thmtools}
\usepackage[inline]{enumitem}
\usepackage{comment}
\usepackage{siunitx}
\usepackage{caption}
\usepackage{hyperref}
\hypersetup{colorlinks=true,allcolors=blue}   
\usepackage[capitalize]{cleveref}
\usepackage{todonotes}
\usepackage[caption=false,font=footnotesize]{subfig}
\usepackage{mathtools}
\usepackage{tikz}
\usepackage{balance}

\newboolean{PGFPlots}
\setboolean{PGFPlots}{true}
\ifthenelse{\boolean{PGFPlots}}{
  \usepackage{pgfplots}
  \usepgfplotslibrary{units,groupplots}
  \usetikzlibrary{external}
  \tikzexternalize
  \tikzset{external/only named=true}
}{}

\newcommand{\tool}{CAPDR\xspace} 
\newcommand{\pdr}{PDR\xspace}
\newcommand{\icthree}{IC3\xspace}
\newcommand{\Inv}{\mathit{Inv}}
\newcommand{\Init}{\mathit{Init}}
\newcommand{\Trans}{\mathit{Trans}}
\newcommand{\Prop}{\mathit{Prop}}
\newcommand{\Reach}{\mathit{Reach}}
\newcommand{\Def}{\mathit{Def}}

\newtheorem{dfn}{Definition}

\newcommand{\maybeincludegraphics}[2][]{%
  \IfFileExists{#2}{\includegraphics[#1]{#2}}{%
    \fbox{\parbox[c][0.18\textheight][c]{0.95\linewidth}{\centering Missing file: \texttt{#2}}}%
  }%
}

\begin{document}

\title{Certificate-Aware Property-Directed Reachability \thanks{This is an extended version of the paper published at the Formal Methods in Computer-Aided Design (FMCAD 2026) conference.}}
\author{
  \IEEEauthorblockN{
    Arman Ferdowsi\textsuperscript{1} and
    Laura Kovács\textsuperscript{2}}\\[0.25em]
  \IEEEauthorblockA{\textsuperscript{1}Faculty of Computer Science, University of Vienna \quad
  \textsuperscript{2}Faculty of Informatics, TU Wien\\
  Email: arman.ferdowsi@univie.ac.at, laura.kovacs@tuwien.ac.at}
}

\maketitle

\begin{abstract}
Property-Directed Reachability (\pdr/\icthree) is widely used for hardware safety verification.
Most implementations optimize time-to-answer, but certified workflows also require compact, inexpensive-to-check, and reproducible certificates.
We introduce \tool, whose ranker is fitted offline on a disjoint development corpus, frozen, and used only to order \pdr-generated proposals on unseen instances to target solving time, certificate size, and checker time.
The ranker remains outside the trusted computing base: it cannot bypass the SAT guards on blocker insertion and pushing; obligation scheduling remains fair, and every SAFE/UNSAFE claim requires independent checker acceptance.
\tool also defines certificate-centric metrics and replay logs for artifact-grade reproducibility.
On the 2024 Hardware Model Checking Competition bit-level safety benchmarks, \tool solves six more instances than the same \pdr engine with ranking disabled.
Across each configuration's checker-accepted solved set, the medians of the certificate-size proxy and checker time decrease by 24.6\% and 49\%, respectively.
After the same post-hoc minimizer is applied to paired SAFE outputs, \tool certificates still contain 16.3\% fewer literals and have 38.5\% lower checker time.
\end{abstract}
\begin{IEEEkeywords}
Hardware model checking, property-directed reachability, proof-carrying verification, certificate checking, learning-guided heuristics
\end{IEEEkeywords}

\section{Introduction}
Property-Directed Reachability (\pdr), also known as \icthree, is one of the most effective SAT-based techniques for proving safety properties of finite-state hardware transition systems~\cite{Bradley2011IC3}.
A key attraction of \pdr is that it naturally produces \emph{independently checkable evidence} (\emph{certificate}):
in the SAFE case, an inductive invariant; in the UNSAFE case, a finite counterexample trace.
Historically, however, this evidence was largely treated as a secondary artifact, while implementations and benchmarks focused on wall-clock time and solved instances.

This emphasis is shifting. The \emph{hardware model checking competition} (HWMCC) introduced mandatory certification for its bit-level safety track in 2024 and continued the workflow in 2025~\cite{FroleyksYuPreinerBiereHeljanko2025CertificatesHWMCC,BiereFroleyksPreiner2025HWMCC}. In essence, certificate-based validation is now a standard mechanism for increasing the trust, reproducibility, and interoperability of verification results, e.g., in software verification~\cite{BeyerStrejcek2025SVComp}. This reflects an underlying principle in automated reasoning: heuristic search need not be trusted when correctness is discharged by an independent checker that validates the emitted certificate. As a consequence, certificates become persistent workflow artifacts rather than merely internal byproducts. They are archived, rechecked, compared across versions, and reused downstream. Consequently, reducing their size and checking cost is crucial, as it significantly lowers workflow overhead.

Once certificates become first-class artifacts, \emph{quality} matters.
Two correct SAFE certificates for the same system may differ by orders of magnitude in size and checking cost.
Oversized invariants slow down or even prevent independent checking. They also complicate reuse in incremental verification and downstream reasoning.
Moreover, modern \pdr implementations are heuristic-driven and may be nondeterministic, due for example to randomized SAT heuristics, clause database effects, and parallelism. As such, \pdr may lead to certificates that vary substantially across runs.
This instability complicates debugging, hinders regression bisection, and reduces confidence in certified workflows.

At the same time, the community has begun exploring machine learning to guide verification,
including learning-guided \pdr/\icthree variants~\cite{HuTangYuZhangZhang2024DeepIC3} and
learning-based algorithm selection for hardware model checking~\cite{LuChienLeeGurfinkelGanesh2025Btor2Select}.
These directions raise a practical requirement: learning should improve end-to-end cost \emph{without} enlarging the trusted computing base.

In this work, we advocate \emph{certificate-aware \pdr}, in short \tool, without changing what must be trusted.
\tool is neither a new certificate format nor a new checker.
It separates \pdr-generated proposals, learned ranking, ordinary \pdr control, and independent validation.
The policy is trained offline on prior instances, then frozen and reused without adaptation across unseen target instances, where it only orders current proposals.
This cross-instance deployment lets the one-time fitting cost be amortized over a recurring stream of related bit-level verification tasks.
Transfer quality is an empirical performance question rather than a soundness assumption.
A poorly transferred policy may waste search effort, but it cannot bypass guards, remove the exact blocker fallback, alter certificate semantics, or cause an unchecked claim.

\noindent
\textbf{Our contributions} are summarized below.  
\begin{itemize}
\item \emph{Certificate-aware objective.} We define a run-level cost model over runtime, certificate size, and independent checker time, turning PDR into a controllable optimization process (\cref{subsec:separation}--\cref{subsec:metrics-objective}).
\item \emph{A soundness-transparent learning interface.} We formalize ranking events, development-time pairwise preferences induced by the run-level objective, and frozen deployment. The policy only orders verifier-generated proposals, while ordinary \pdr guards and the independent checker preserve correctness (\cref{subsec:pdr-guards}--\cref{subsec:policy-interface}).
\item \emph{Replay logs for reproducibility.} We propose a compact replay log that records nondeterministic choices together with the solver-derived artifacts needed to determinize re-execution (\cref{subsec:policy-interface}--\cref{subsec:capdr-alg}).
\item \emph{Soundness} of \tool is established, in a policy-independent manner (\cref{sec:soundness}).
    \item \emph{Experimental evaluation under a disjoint development/evaluation protocol.} On certified hardware benchmarks, frozen-policy evaluation on held-out instances reduces certificate size and checker time while improving solved coverage. We additionally apply one shared implementation of the post-hoc minimal safe inductive subset (MSIS) method of~\cite{IvriiGurfinkelBelov2014SmallInvariants} to paired baseline and \tool SAFE invariants, showing that the online-search advantage survives equal postprocessing (\cref{sec:experiments}).
\end{itemize}

\noindent
\emph{Paper organization:}
\cref{sec:RW} reviews related work. \cref{sec:background} recalls PDR and certificate checking. \cref{sec:capdr} presents \tool. \cref{sec:soundness} proves checker-relative soundness, and \cref{sec:experiments} evaluates the approach.

\section{Related Work}
\label{sec:RW}
\subsection{Certification in hardware model checking}
Certification in hardware model checking predates HWMCC'24, including witness-circuit certification and stratified certification for k-induction~\cite{YuBiereHeljanko2021Progress,YuFroleyksBiereHeljanko2022Stratified}. HWMCC'24 made certificates mandatory for the bit-level safety track~\cite{FroleyksYuPreinerBiereHeljanko2025CertificatesHWMCC}, and HWMCC'25 continued this workflow while requiring certificates for satisfiable bit-level liveness results~\cite{BiereFroleyksPreiner2025HWMCC}. Certification is not limited to the core model checking algorithm: certifying preprocessing has been developed for phase abstraction~\cite{FroleyksYuBiereHeljanko2024PhaseAbstraction}, and compositional certification frameworks show how to combine certificates across decomposed problems~\cite{YuFroleyksBiereHeljanko2023CompositionalCertification}. \tool is orthogonal to these efforts: rather than proposing a new certificate format, we make the \emph{heuristic search} inside \pdr certificate-aware and explicitly optimize for downstream checking cost.

A closely related direction studies how to extract, certify, or minimize invariants once they have been found, including minimal safe inductive subsets and proof-certificate extraction for model checkers~\cite{IvriiGurfinkelBelov2014SmallInvariants,MebsoutTinelli2016ProofCertificates,GriggioRoveriTonetta2021CertifyingSAT}. Ivrii et al.~\cite{IvriiGurfinkelBelov2014SmallInvariants} start from an already obtained conjunctive normal form (CNF) invariant and search for a subset-minimal safe inductive subset. In contrast, \tool changes the online \pdr search, so it can discover different lemmas and reach a different fixed point before any post-fixpoint reduction is applied. The two approaches are therefore orthogonal rather than substitutes. We evaluate both a lightweight checker-preserving greedy pass (\textsc{inv-min}) and a fair paired comparison in which the same combined MSIS implementation of~\cite{IvriiGurfinkelBelov2014SmallInvariants} is applied to baseline and \tool invariants.
\subsection{Strengthening and optimizing \pdr}
Since the first version of \icthree~\cite{Bradley2011IC3}, performance improvements have been made  both at the algorithm layer and at the SAT-solver interface.
For example, removing activation literals from single-clause assumptions can speed up \pdr-style workloads~\cite{FroleyksBiere2021SingleClauseAssumption}.
Recent work has also investigated lemma-quality criteria and strategies that improve the ability to push and reuse learned information, e.g., via searching for \emph{i-good} lemmas~\cite{XiaBecchiCimattiGriggioLiPu2023iGood}.
More recently, extended-resolution reasoning has been integrated into \pdr to improve proof power~\cite{LukaVizel2025PdrER}, and dedicated SAT-solver optimizations tailored to \icthree workloads have been shown to yield substantial speedups~\cite{SuYangCiLiBuHuang2025DeeplyOptimizing}. Another relevant line augments \icthree with internal signals, allowing invariants
over an enlarged variable vocabulary when no concise latch-only invariant exists~\cite{dureja2021ic3}.
This line is particularly relevant to our optional auxiliary-variable modes:
unlike such approaches, \tool\ makes the checker-facing certificate interface explicit by
requiring either an $X$-only certificate accepted by the base checker or an extended
certificate accompanied by a conservative auxiliary-state relation. \tool can incorporate such techniques as optional guarded actions, enabling controlled trade-offs between solver runtime and certificate-centric metrics.

\subsection{Learning-guided verification}
DeepIC3~\cite{HuTangYuZhangZhang2024DeepIC3}, learned lemma prediction~\cite{SuYangCi2024PredictingLemmas}, LeGend~\cite{miao2026legend}, and A-IC3~\cite{ZhouHuZhangZhang2026AIC3} use learning to improve \pdr/\icthree search. Btor2-Select~\cite{LuChienLeeGurfinkelGanesh2025Btor2Select} applies learning at the algorithm-selection level, while IC3-Evolve~\cite{miao2026ic3} evolves solver heuristics offline and admits results through independently checked proofs or witnesses. The cited evaluations principally target solving performance, lemma prediction, or strategy selection. They do not report the same combination of checker-accepted CNF-invariant size, independently measured checker time, and replay fidelity under a common certificate interface. A-IC3 also updates a bandit policy from online feedback, whereas \tool fits its ranker only on the development corpus and freezes it for deployment.

Our headline baseline is our implementation of Bradley's \pdr algorithm~\cite{Bradley2011IC3}. Within the paired experiments, baseline \pdr and \tool share the solver-facing core, preprocessing, resource limits, certificate emission, and checker. They differ only in the order of verifier-generated choices, together with the separately labeled optional post-fixpoint pass. This controlled internal comparison isolates the effect of certificate-aware guidance within our engine. We do not claim that the baseline is the strongest available HWMCC \pdr implementation. A total-tool comparison with learning-guided systems remains informative for coverage, but it cannot isolate ranking because preprocessing, SAT tuning, lemma synthesis, invariant emission, and checking infrastructure also differ.

\subsection{Proof formats and trustworthy checking}
Trustworthy reasoning has a long history in SAT and satisfiability modulo theories (SMT).
DRAT proofs and the DRAT-trim checker~\cite{WetzlerHeuleHunt2014DRATTrim} provided an expressive and practical proof format for SAT,
and LRAT adds hints that enable efficient checking in trusted systems~\cite{CruzFilipeHeuleHuntKaufmannSchneiderKamp2017LRAT}.
Efficient LRAT checking integrated with modern SAT solvers is demonstrated in~\cite{PollittFleuryBiere2023FasterLRAT}, and formally verified checking has also been developed~\cite{Lammich2024FastVerifiedUNSAT}.

Recent work extends clausal-proof certification to Lean, importantly via a verified Lean checker~\cite{codel2024verified} for substitution-redundancy proofs; this yields a clausal proof system generalizing propagation redundancy (PR) and resolution asymmetric tautology (RAT).
For SMT, Alethe~\cite{SchurrFleuryBarbosaFontaine2021Alethe} and tools like Carcara~\cite{AndreottiLachnittBarbosa2023Carcara}, as well as proof reconstruction into Isabelle/HOL~\cite{LachnittEtAl2025IsabelleSMTReconstruction} and Lean~\cite{mohamed2025lean}, show how solver outputs can be independently validated.
Extending DRAT-style proofs to SMT further highlights the broader relevance of proof logging~\cite{HitarthCodelLachnittDutertre2024eDRAT}.
\tool complements these efforts by focusing on certifying \emph{model checking outcomes} (invariants/traces) and by making the search procedure certificate-aware while keeping the correctness argument entirely checker-based.

\section{Background}
\label{sec:background}
We briefly recall the transition-system model, the SAFE/UNSAFE certificate view adopted in this paper, and the \pdr/\icthree and certificate notions used later.

\subsection{Transition systems and safety}
We consider finite-state transition systems over Boolean state variables $X$ (latches)
and Boolean input variables $U$.
A system is specified by a triple $(\Init(X), \Trans(X,U,X'), \Prop(X))$, where
$\Init$ characterizes initial states, $\Trans$ is the transition relation,
and $\Prop$ is a safety property that must hold in all reachable states.

\begin{dfn}[Reachable state]
A state $x$ is \emph{reachable} if there exists a finite sequence of inputs $u_0,\dots,u_{k-1}$ and states
$x_0,\dots,x_k$ such that $\Init(x_0)$, $\Trans(x_i,u_i,x_{i+1})$ for all $i<k$, and $x_k=x$. We denote the set of reachable states by $\Reach$.
\end{dfn}

\begin{dfn}[Safety]
The system is \emph{safe} if and only if every reachable state
satisfies the safety property. Formally, $\forall x.\;\bigl(x\in\Reach \Rightarrow \Prop(x)\bigr)$.
\end{dfn}

\subsection{Certificates}
\label{subsec:certificates}
A standard way to certify safety-verification outcomes is via explicit, independently checkable artifacts: an \emph{inductive invariant} in the SAFE case and a \emph{counterexample} trace in the UNSAFE case.

\begin{dfn}[Safe certificate]
A \emph{safe certificate} is a Boolean formula $\Inv(X)$ over the state variables satisfying:
\begin{align}
  \Init(X) &\Rightarrow \Inv(X) \label{eq:init} \\
  \Inv(X) \wedge \Trans(X,U,X') &\Rightarrow \Inv(X') \label{eq:step} \\
  \Inv(X) &\Rightarrow \Prop(X) \label{eq:prop}
\end{align}
\end{dfn}

In \emph{base mode}, where the certificate uses only the original state variables \(X\), checking~\cref{eq:init}-\cref{eq:prop} reduces to three
SAT queries (unsatisfiability checks), and is typically much cheaper than
re-running model checking from scratch.
Competition-grade certification frameworks often use more structured encodings of invariants
(e.g., witness circuits) and establish the certificate's checkability and semantics
in a dedicated checker~\cite{FroleyksYuPreinerBiereHeljanko2025CertificatesHWMCC}.

\begin{dfn}[Unsafe certificate]
\label{def:UnsafeCer}
An \emph{unsafe certificate} is a finite trace $x_0,\dots,x_k$ with inputs $u_0,\dots,u_{k-1}$
such that $\Init(x_0)$, each step satisfies $\Trans(x_i,u_i,x_{i+1})$, and $\neg \Prop(x_k)$.
This can be checked by simulation or by SAT.
\end{dfn}

\subsection{Property-Directed Reachability (PDR/\icthree)}
\pdr incrementally constructs a sequence of clause sets (``frames'') $F_0, F_1, \dots$ such that:
(i) $F_0$ is initialized from $\Init$, (ii) each $F_i$ over-approximates the states reachable within $i$ steps,
and (iii) the sequence is monotone ($F_i \Rightarrow F_{i+1}$).
The \pdr algorithm searches for \emph{counterexamples to induction} (CTIs) at successive frames and blocks them by learning clauses that exclude the offending state while remaining inductive relative to the previous frame.
Clauses that satisfy the corresponding push condition can be propagated forward, strengthening later frames.
When, after blocking and pushing, two consecutive frames coincide, their common clause set is an inductive invariant and therefore a SAFE certificate.
If blocking reaches the initial frame, the predecessor chain yields a concrete counterexample trace; see~\cite{Bradley2011IC3,Bradley2012UnderstandingIC3}.

Modern \pdr implementations differ primarily in how they perform generalization,
how they order proof obligations, which are pending blocking tasks, and how they manage incremental SAT queries; a classic implementation-oriented treatment is~\cite{EenMishchenkoBrayton2011EfficientPDR}.

\section{Certificate-Aware PDR}
\label{sec:capdr}
In this section, we present \tool as a certificate-aware variant of \pdr/\icthree that, by adding an offline learning ranker, jointly targets time-to-answer, certificate size, and checker cost without adding the learned ranker to the trusted computing base.
It ranks PDR-generated choices at four choice points (CPs), namely blocker generalization (CP1), obligation ordering (CP2), clause pushing (CP3), and optional guarded extensions (CP4), defined with their control conditions in \cref{subsec:pdr-guards,subsec:choice-points}.
\cref{fig:capdr-overview} shows the division between learned ranking, guarded \pdr control, and independent checking.

\begin{figure*}[t]
  \centering
  \includegraphics[width=0.99\textwidth]{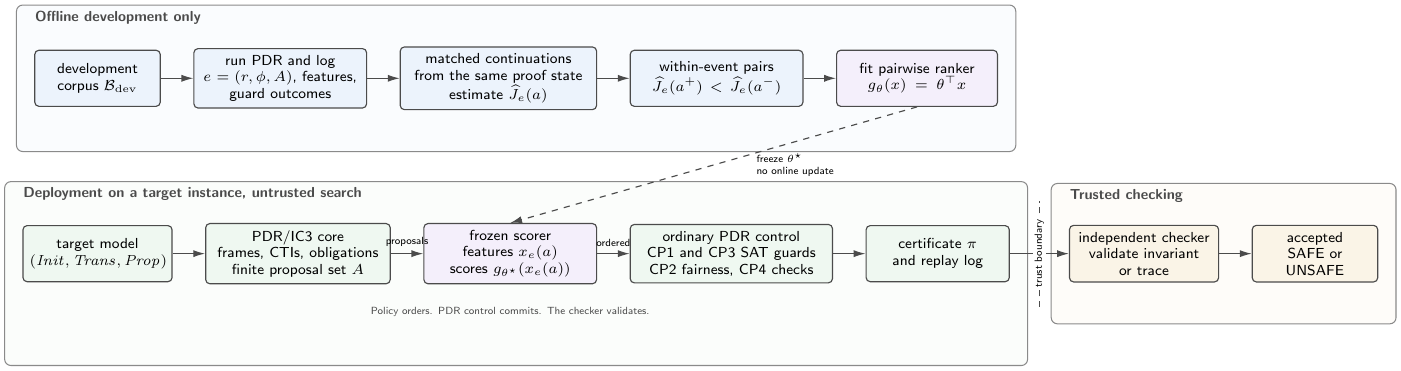}
  \caption{\footnotesize \tool overview and trust boundary. The ranker is fitted offline on the development corpus, frozen once, and reused without adaptation on unseen target instances. It only orders verifier-generated proposals. CP1 and CP3 remain SAT-guarded, CP2 remains subject to fairness, CP4 retains conservative validation, and the independent checker validates the final artifact before SAFE/UNSAFE is reported.}
  \label{fig:capdr-overview}
\end{figure*}

Offline \emph{development} uses guarded continuations and the run-level objective over solving time, certificate size, and checker cost to derive pairwise preferences and fit a ranker on prior instances.
Replay metrics measure reproducibility
(\cref{subsec:metrics-objective,subsec:policy-interface}).
The fitted ranker is then frozen and reused unchanged on unseen targets.

During \emph{deployment}, at each choice point the verifier generates a finite proposal set, the frozen ranker orders it, and ordinary \pdr control governs commitment.
CP1 and CP3 change frames only after SAT guards succeed, CP2 selects the next obligation subject to fairness, and CP4 retains conservative validation.
Deployment performs only feature extraction, scoring, and ordering, with no online objective evaluation, rollout, or per-proposal checker call.
SAFE/UNSAFE is reported only after checker acceptance.

%

\subsection{Separation and trust boundary}
\label{subsec:separation}
In \tool, learning may affect search but never the correctness argument. The \emph{trusted} computing base consists only of a standalone certificate checker: its parsing layer, the SAT engine used for SAFE certificates, the trace-evaluation/simulation routine used for UNSAFE certificates, and any validation routine for the auxiliary-variable extension mode defined in \cref{subsec:cert-semantics}. Everything else, including the possibly nondeterministic certificate-producing verifier, the learned policy, feature extraction, training, and hyperparameter tuning, is \emph{untrusted}.

Accordingly, all correctness claims are conditioned solely on checker acceptance. The policy may rank only alternatives already admitted by the underlying algorithm, and any state-changing action, such as learning or pushing a clause, is committed only after the corresponding \pdr-style satisfiability guard succeeds, or its conservative auxiliary-variable variant when auxiliaries are present. Thus, learning can change the search order and which valid lemmas are found, but it cannot by itself make an invalid certificate cross the trust boundary.

\subsection{Formal certificate semantics}
\label{subsec:cert-semantics}
In base mode, a safe certificate is emitted as a CNF invariant $\Inv(X) \;\equiv\; \bigwedge_{C \in \Inv} C(X)$,
where each $C$ is a clause (a disjunction of literals over $X$).
Its semantics is exactly that of \cref{eq:init,eq:step,eq:prop}. An unsafe certificate is emitted as a finite trace $
  \pi \;=\; (x_0,u_0,x_1,\dots,u_{k-1},x_k)$
with $x_i \in \{0,1\}^{|X|}$ and $u_i \in \{0,1\}^{|U|}$.
Its semantics is exactly that of \cref{def:UnsafeCer}.

Optional strengthening may introduce auxiliary state bits $Y$. To preserve soundness for the original system over $X$, \tool supports either \emph{projection mode}, which emits an $X$-only candidate invariant later validated by the base SAFE checker, or \emph{extension mode}, which emits $(\Def,\Inv(X,Y))$. Writing $\Def'(X',Y')$ for the primed copy of $\Def(X,Y)$, define
\begin{align}
\Init^{\mathsf{ext}} &\triangleq \Init(X)\wedge \Def(X,Y), \notag\\
\Trans^{\mathsf{ext}} &\triangleq \Trans(X,U,X')\wedge \Def(X,Y)
\wedge \Def'(X',Y'). \notag
\end{align}
We require $\Def$ to be conservative under projection to $X$:
\begin{align}
\forall x.\quad \Init(x) &\Rightarrow \exists y.\;\Def(x,y),
\label{eq:defext-init}\\
\forall x,u,x',y.\;
\Trans(x,u,x')\wedge \Def(x,y)
&\Rightarrow \exists y'.\;\Def(x',y').
\label{eq:defext-step}
\end{align}

In extension mode, the checker validates both $\Def$ and $\Inv$ on the extended system. Acceptance then implies safety of the original system (\cref{prop:defext-sound}). Since the experiments in \cref{sec:experiments} use base mode, here we restrict our attention to the $X$-only case. More details about auxiliary extensions are deferred to Appendix~\ref{app:aux-extension}.

\subsection{PDR core invariants and guard checks}
\label{subsec:pdr-guards}
To specify the learning interface precisely, we restate the frame invariants and SAT guard checks that determine when a proposed state-changing action may be committed in \pdr.
Recall that a frame $F_i$ is a CNF over $X$.
At any time, the engine maintains a \emph{finite} prefix of frames
$F_0,\dots,F_{k+1}$ for a current depth $k\ge 0$. We call level \(k\) the current frontier. New frames are created lazily
and initialized to $\Prop$. For the algorithmic presentation, we assume that $\Prop$ has been fixed once in CNF (or via a fixed clausal encoding shared by the verifier and the checker); by abuse of notation, we write $\Prop$ for that clausal form.

Since safety properties may be violated already in the initial states,
\tool first checks whether $\Init(X)\wedge \neg\Prop(X)$ is satisfiable.
If it is, the system is immediately \textsc{unsafe} with a length-$0$ trace,
and the run terminates with an UNSAFE certificate accepted by the independent checker.
Hence, for the remainder we may assume $\Init \Rightarrow \Prop$.

We maintain the standard \pdr frame invariants for all currently allocated indices:
\begin{align}
  &F_0 \equiv \Init \label{eq:frame0} \\
  &F_i \Rightarrow F_{i+1} \quad \forall \ 0\le i \le k \label{eq:framemon} \\
  &F_i \Rightarrow \Prop \quad \forall \ 1\le i \le k{+}1 \label{eq:frameprop} \\
  &F_i(X) \wedge \Trans(X,U,X') \Rightarrow F_{i+1}(X') \quad \forall \ 0\le i < k.
    \label{eq:framestep}
\end{align}
\cref{eq:framestep} holds at every \emph{closed} level $i<k$, meaning that consecution from \(F_i\) to \(F_{i+1}\) has been established.
At the frontier level \(k\), ruling out frontier CTIs establishes
\(F_k(X)\wedge \Trans(X,U,X') \Rightarrow \Prop(X')\), and any strengthening
of \(F_{k+1}\) is justified either by a blocker check relative to \(F_k\)
or by a successful push check.
Hence, when after the push phase $F_k \equiv F_{k+1}$, the common frame is inductive and thus a SAFE certificate.

\begin{dfn}[SAT encodings]
All implications are checked via unsatisfiability: $A \Rightarrow B \ \text{iff} \ A \wedge \neg B \text{ is UNSAT}$.
We define and write $\mathsf{UNSAT}(\Phi)$ and $\mathsf{SAT}(\Phi)$ for \emph{(un)satisfiability
queries} to the underlying SAT backend on a fixed equisatisfiable clausal
encoding of $\Phi$; when $\Phi$ is already in CNF, this is just $\Phi$ itself.
\end{dfn}

\begin{dfn}[State-cube convention]
A \emph{state cube} is a complete \emph{cube} over \(X\), i.e., a conjunction containing exactly one literal for each variable in \(X\).
Such a cube represents exactly one concrete state.
For a state cube \(d\), its primed copy \(d(X')\) is the same conjunction over \(X'\), constraining the next-state variables to that state, and \(\neg d\) is the clause excluding that state.
In the obligation mechanism below, \(d\) and \(d_{\mathsf{pred}}\) always denote state cubes extracted from SAT models.
\end{dfn}

A \emph{proof obligation} is a pair \((i,d)\) where
\(i \ge 0\) and \(d\) is a state cube over \(X\) representing one concrete target state. Equivalently, \(d(X')\) fixes the next-state variables to that state. Each non-root obligation also stores the successor cube and input valuation by which it was generated. When an initial cube is reached, following these links yields a concrete suffix ending in a bad state.

To find a bad successor at depth $i$, the engine checks
\begin{equation}
  \mathsf{SAT}\Bigl(F_i(X) \wedge \Trans(X,U,X') \wedge \neg\Prop(X')\Bigr).
  \label{eq:cti-query}
\end{equation}
If \cref{eq:cti-query} is satisfiable, a model yields a current-state valuation on \(X\) and a bad successor valuation on \(X'\).
The engine unprimes the \(X'\)-literals to obtain the bad-successor state cube \(d\) over \(X\) and creates the obligation \((i{+}1,d)\).

If \(i=0\), blocking has reached the initial frame and the stored
obligation chain is emitted as a candidate UNSAFE trace. More generally,
before learning any blocker for an obligation \((i,d)\), the engine
checks whether $\mathsf{SAT}\bigl(\Init(X)\wedge d(X)\bigr)$.
If this query is SAT, then, because \(d\) is a complete cube, the target
state represented by \(d\) is an initial state. The stored successor
links therefore yield a candidate UNSAFE trace, and no blocker is learned
for \(d\). Only when \(i>0\) and
\(\mathsf{UNSAT}(\Init(X)\wedge d(X))\) does \pdr attempt to block
\((i,d)\) by the predecessor query:
\begin{equation}
  \mathsf{SAT}\Bigl(F_{i-1}(X) \wedge \Trans(X,U,X') \wedge d(X')\Bigr).
  \label{eq:pred-query}
\end{equation}
If \cref{eq:pred-query} is SAT, the model yields a concrete predecessor
valuation on \(X\) together with a witness input valuation on \(U\); let
\(d_{\mathsf{pred}}\) be the corresponding predecessor state cube over
\(X\). We store the successor link
\(d_{\mathsf{pred}}\xrightarrow{u} d\), requeue \((i,d)\), and add
\((i{-}1,d_{\mathsf{pred}})\).
If \cref{eq:pred-query} is UNSAT, then no state satisfying \(F_{i-1}\)
can reach the target cube \(d\) in one step. The fallback
(\emph{base}) PDR blocker is
\(C_0\triangleq \neg d\), the clause excluding exactly the concrete state
represented by \(d\). More generally, \tool may try a generalized blocker,
but any clause \(C(X)\) actually inserted for obligation \((i,d)\) must
satisfy the same guard conditions:
\begin{align}
  d(X) &\Rightarrow \neg C(X) \label{eq:block} \\
  \Init(X) &\Rightarrow C(X) \label{eq:init-block} \\
  F_{i-1}(X) \wedge \Trans(X,U,X') &\Rightarrow C(X').
    \label{eq:relind}
\end{align}
These conditions express target blocking, initiation, and relative inductiveness, respectively.
Condition~\eqref{eq:init-block} is equivalent to
\(\mathsf{UNSAT}(\Init(X)\wedge \neg C(X))\), and
condition~\eqref{eq:relind} is equivalent to $\mathsf{UNSAT}\bigl(
    F_{i-1}(X)\wedge \Trans(X,U,X')\wedge \neg C(X')
  \bigr)$.

\medskip
\noindent
\emph{Clause insertion:}
When \cref{eq:block,eq:init-block,eq:relind} hold, the verifier may conjoin
\(C\) to \(F_1,\dots,F_i\), never to \(F_0\). This preserves monotonicity,
closed-level consecution, and \(F_j\Rightarrow\Prop\) by the standard guarded
insertion argument given in Appendix~\ref{app:frame-preservation}. Thus the
policy only chooses which candidate \(C\) to try; insertion itself remains
SAT-guarded.

\subsection{Learning-guided choice points}
\label{subsec:choice-points}
The guarded control above leaves four choice points, at which the verifier must order a finite set of candidates.
Candidate-set membership asserts only syntactic and algorithmic admissibility. It does not mean that a clause insertion or push has passed its commitment guard. The policy supplies only an order.

\noindent
\emph{\textbf{CP1} Learning-guided blocker generalization.}
After an obligation \((i,d)\) passes the initial-state check and its predecessor query is UNSAT, the exact fallback blocker is \(C_0\triangleq\neg d\).
The verifier may also construct generalized blockers
\[
  \mathcal{C}(d,i)\subseteq
  \{\neg d^\star\mid d^\star\subseteq d\},
  \qquad C_0\in\mathcal{C}(d,i),
\]
where \(d^\star\) is obtained by dropping literals from the complete target cube \(d\).
Because \(d\Rightarrow d^\star\), every candidate \(C=\neg d^\star\) automatically satisfies the blocking condition \cref{eq:block}.
A smaller subcube yields a shorter and logically stronger clause, but it may require more SAT effort or fail initiation or relative inductiveness.
The implementation seeds the finite set with an UNSAT-core-derived subcube when available and explores additional subcubes by deterministic budgeted literal deletion.
The learned policy orders this set, and the verifier commits the first candidate that satisfies \cref{eq:init-block,eq:relind}.
If all proper generalizations fail, \(C_0\) remains available.

As an illustrative example, consider why learned ranking is not reducible to shortest-first ordering. Let the target cube be $d=x_1\wedge\neg x_2\wedge x_3$.
The verifier-generated menu may contain the one-literal clause $C_{\mathsf{s}}=\neg x_3$, the two-literal clause $C_1=x_2\vee\neg x_3$,
and the exact fallback $C_0=\neg d=\neg x_1\vee x_2\vee\neg x_3$.
All three block the target cube, and $C_{\mathsf{s}}\Rightarrow C_1\Rightarrow C_0$.
Thus, ordering candidates by increasing length always tests the strongest generalization \(C_{\mathsf{s}}\) first, although it may fail initiation or relative inductiveness in the current frame. \tool may instead test \(C_1\) first when the available context indicates that it is more likely to pass the guards. Such context includes frame information, recency, blocker and push history, and SAT-side statistics. Even when a shorter candidate passes, length alone does not determine subsequent obligation generation, push success, convergence, or final checker time. A candidate is inserted only after
\[
  \mathsf{UNSAT}(\Init\wedge\neg C),
  \qquad
  \mathsf{UNSAT}(F_{i-1}\wedge\Trans\wedge\neg C')
\]
succeed. If every proper generalization fails, \(C_0\) remains available by \cref{prop:base-blocker}.

Clause length is one CP1 feature, not a replacement for the policy. CP2 ranks proof obligations and has no clause-length notion. A CP3 push opportunity contains a clause, but its usefulness also depends on the source and target frames, previous push outcomes, and the effect on convergence. Hence no single shortest-formula rule covers the three learned choice points.

\begin{restatable}[Base blocker is admissible after the initial-state check]{prp}{BaseBlockerAdmissible}
\label{prop:base-blocker}
If \(i>0\), \(\mathsf{UNSAT}(\Init(X)\wedge d(X))\), and the predecessor
query~\cref{eq:pred-query} for an obligation \((i,d)\) is UNSAT, then
the base clause \(C_0 \triangleq \neg d\) satisfies
\cref{eq:block,eq:init-block,eq:relind}.
\end{restatable}

\begin{proof}
Let \(C_0 \triangleq \neg d\).
Since \(d\) is a state cube over \(X\), we have
\(d(X) \Rightarrow \neg C_0(X)\), so \(C_0\) satisfies the blocking
condition~\cref{eq:block}.

Since \(\mathsf{UNSAT}(\Init(X)\wedge d(X))\), we have
\(\Init(X)\Rightarrow \neg d(X)\). Because \(C_0=\neg d\), this is
exactly the initiation condition~\cref{eq:init-block}.
For the relative-inductiveness guard, consider
\[
\mathsf{UNSAT}\bigl(F_{i-1}(X)\wedge \Trans(X,U,X') \wedge \neg C_0(X')\bigr).
\]
Because \(C_0(X')=\neg d(X')\), this is equivalent to
\[
\mathsf{UNSAT}\bigl(F_{i-1}(X)\wedge \Trans(X,U,X') \wedge d(X')\bigr),
\]
which is exactly the predecessor query~\cref{eq:pred-query}, assumed
UNSAT. Hence \(C_0\) satisfies~\cref{eq:relind}.

Finally, \(\neg d\in\mathcal{C}(d,i)\) by construction, and the blocking
loop scans the finite ordered candidate list until it finds a candidate
satisfying the initiation and relative-inductiveness guards. Since
\(C_0\) satisfies both, the loop must accept some candidate regardless of
the policy ranking.
\end{proof}

\noindent
\emph{\textbf{CP2} Obligation ordering.}
Here the candidate set comprises the pending obligations currently eligible for processing, which we call ready obligations.
The policy selects which obligation \((i,d)\) is processed next.
This scheduling decision does not by itself modify any frame, although it can change which predecessors and blockers are discovered later.
A fairness rule such as bounded starvation ensures that a low-scored obligation cannot be postponed forever.
Useful features include the frame index, recency, requeue count, predecessor-chain depth, and SAT-side history.

\noindent
\emph{\textbf{CP3} Pushing schedule and clause selection.}
For a learned clause \(C\in F_i\) with \(i\geq 1\), a push attempt checks
\begin{equation}
  F_i(X)\wedge\Trans(X,U,X')\Rightarrow C(X'),
  \label{eq:push-check}
\end{equation}
equivalently \(\mathsf{UNSAT}(F_i\wedge\Trans\wedge\neg C')\).
CP3 ranks eligible push attempts and does not synthesize new clauses.
A clause is added to \(F_{i+1}\) only when \cref{eq:push-check} holds.
The order can affect convergence and the fixed-point frame because successful pushes change the frames seen by later scheduling decisions.

\noindent
\emph{\textbf{CP4} Optional guarded extensions.}
Optional strengthening methods may be exposed as policy choices only when their clauses obey the appropriate commitment guards and any auxiliary variables have explicit checker-facing semantics as in \cref{subsec:cert-semantics}.
CP4 is disabled in the base-mode experiments.

Thus every choice point follows the same interface. The verifier generates proposals, the policy orders them, and the appropriate guarded or fair control decides what is committed.

\subsection{Run-level cost metrics and objective}
\label{subsec:metrics-objective}
Given the guard-checked choices above, we evaluate a run by the vector $(t,\mathrm{size}(\pi),t_{\mathsf{chk}})$, where \(t\) is solving time, \(\pi\) is the emitted certificate, and \(t_{\mathsf{chk}}\) is the time taken by an independent checker on the emitted certificate. Since checking cost depends on the checker and SAT engine, we treat \(t_{\mathsf{chk}}\) as the ground-truth downstream cost.
For learning, one may additionally log proxies such as the number of SAT calls and their sizes. Below, we define the certificate-size component \(\mathrm{size}(\pi)\), and then use a scalarization only when a single optimization target is needed for learning or end-to-end reporting.

\begin{dfn}[Certificate-size proxy]
For a CNF invariant \(\Inv\), let $\mathrm{lit}(\Inv) \triangleq \sum_{C\in\Inv} |C|$.
For a trace \(\pi=(x_0,u_0,x_1,\dots,u_{k-1},x_k)\), let $\mathrm{bits}(\pi) \triangleq (k{+}1)\cdot|X| + k\cdot|U|$.
We define
\[
  \mathrm{size}(\pi) \triangleq
  \begin{cases}
    \mathrm{lit}(\Inv), & \text{if } \pi=\Inv \text{ (SAFE)},\\
    \mathrm{bits}(\pi), & \text{if } \pi \text{ is a trace (UNSAFE)}.
  \end{cases}
\]
\end{dfn}

Literal count is a syntactic size proxy for the emitted base-mode CNF invariant, not the sole certificate-quality measure and not a claim of semantic minimality. Equivalent invariants can have different CNF representations. Clause count, the clause-width distribution, the size of the checker encoding, and witness-circuit size are also meaningful. We use literal count because it is defined directly on the emitted CNF and is independent of the checker implementation once that artifact is fixed. The primary downstream measure is nevertheless the independently measured checker time \(t_{\mathsf{chk}}\). The experimental conclusions are supported by both quantities rather than by literal count alone.

\begin{dfn}[Scalarized objective]
When a single scalar is needed, we combine the run-level vector \((t,\mathrm{size}(\pi),t_{\mathsf{chk}})\) as
\begin{equation}
  J \;=\; \alpha \cdot t \;+\; \beta \cdot \mathrm{size}(\pi) \;+\; \gamma \cdot t_{\mathsf{chk}},
  \label{eq:objective-revised}
\end{equation}
with user-chosen weights \(\alpha,\beta,\gamma \ge 0\).
Thus \(J\) is simply a tunable scalarization of solving time, certificate bulk, and downstream checking effort.
\end{dfn}

In the evaluated prototype, reproducibility is not folded into \(J\).
Cross-seed stability and replay fidelity are reported separately in \cref{subsec:policy-interface}.

The proxy \(\mathrm{size}(\pi)\) intentionally mixes two outcome-specific measures:
literal count for SAFE certificates and trace bit count for UNSAFE certificates.
Accordingly, aggregate SAFE+UNSAFE summaries using \(\mathrm{size}(\pi)\) (or a scalarization that includes it)
should be read only as coarse workflow-level bookkeeping across outcome types.
Our main certificate-quality analysis is therefore outcome-homogeneous, focusing on SAFE invariants in
\cref{tab:ablation,fig:tradeoff,fig:checker-cdf,fig:stability}.

We still report the underlying vector metrics \((t,\mathrm{size}(\pi),t_{\mathsf{chk}})\). The local scorer introduced in \cref{subsec:policy-interface} is used only as a surrogate for the action values induced by this global objective.

\subsection{Policy learning, deployment, and reproducibility}
\label{subsec:policy-interface}

\noindent
\emph{Ranking events and features.}
At a choice point CP$_r$, the current proof state \(\phi\) induces a ranking event
\[
  e=(r,\phi,A),
\]
where \(A\) is the finite set of blocker proposals for CP1, ready obligations for CP2, eligible push attempts for CP3, or guarded extension proposals for CP4.
For each \(a\in A\), the verifier computes a deterministic feature vector
\(x_e(a)\triangleq\psi(r,\phi,a)\). In the evaluated prototype, \(\psi\) uses verifier-available attributes including frame index, candidate length where applicable, recency, blocker or push history, and SAT-side statistics. The scorer assigns \(g_\theta(x_e(a))\), with higher values placed earlier in the order.
At CP1 and CP3, selecting \(a\) means trying that proposal first while preserving all guards.
It never means forcing an invalid state change.

\noindent
\emph{Development-time pairwise supervision.}
The objective \(J\) is a delayed run-level quantity, so it cannot be read directly from the local proof state.
Development therefore estimates a continuation cost \(\widehat J_e(a)\) for selected alternatives at the same event.
A continuation begins from the same saved proof state, places \(a\) first, and then follows a fixed continuation policy with the usual guards, fallback rules, and fairness constraints. Matched alternatives use the same remaining time and memory budget, solver configuration and seed, and post-action tie-breaking. The first proposal is the only intentional intervention. The estimate may come from a measured development continuation or from an explicit offline rollout. A pairwise label is formed only when two alternatives from the same event have comparable estimates under this common protocol. Missing or censored continuations are omitted unless they receive the explicit failure penalty defined in Appendix~\ref{app:local-objective}. This within-event construction avoids comparing action costs across unrelated proof states or benchmark instances.
Writing \(a^+\succ_e a^-\) for the preferred alternative, the label is
\[
  a^+\succ_e a^-
  \quad\Longleftrightarrow\quad
  \widehat J_e(a^+)<\widehat J_e(a^-).
\]
For the linear ranker used in \cref{sec:experiments}, let
\(x^+=x_e(a^+)\), \(x^-=x_e(a^-)\), and \(\Delta x=x^+-x^-\).
Then
\begin{align*}
  g_\theta(x)&=\theta^\top x,\\
  \ell_\theta(a^+,a^-)
  &=\log\!\left(1+\exp\!\left[-\theta^\top\Delta x\right]\right).
\end{align*}
Minimizing the sum of these losses with \(L_2\) regularization makes a lower-cost alternative receive a higher score.
The complete objective and failure-penalty convention are given in Appendix~\ref{app:local-objective}.

\noindent
\emph{Frozen deployment.}
After fitting, the learned parameter vector \(\theta\) is frozen and reused unchanged across target instances.
For every event on a target instance, \tool computes the current features, evaluates the scorer, and applies fixed tie-breaking to obtain a deterministic total order.
It does not compute \(J\) or \(\widehat J_e\).
It performs no rollout, parameter update, or per-candidate checker call.
CP1 and CP3 scan the score order until a guard-valid proposal is found.
CP2 selects the highest-ranked ready obligation subject to fairness.
CP4 remains subject to its conservative validation conditions.
The final checker is invoked only on the emitted invariant or trace.

\noindent
\emph{Applicability and transfer.}
The features describe local \pdr contexts and candidates rather than benchmark identities, so the frozen ranker can be reused on previously unseen systems and its one-time fitting cost can be amortized over repeated certified verification.
Effective transfer is plausible when development and target runs induce similar local decision distributions, but it is not guaranteed. Nevertheless, although distribution shift may reduce efficiency or certificate quality, soundness is not affected because guards, fairness, the exact CP1 fallback, and checker acceptance remain unchanged.

\noindent
\emph{Replay log.}
To support artifact-grade reproducibility, the engine records tool and SAT-solver versions, configurations, seeds, solver-derived objects that affect control flow, the attempted actions and guard outcomes, and hashes of the relevant contexts.
A replay reuses the recorded objects directly, or checks them against a determinized recomputation.
A hash or guard-outcome mismatch signals divergence caused by nondeterminism or implementation drift.

\noindent
\emph{Stability and replay-fidelity metrics.}
For SAFE certificates, let \(\mathsf{canon}(\Inv)\) denote the clause set obtained by sorting literals within each clause and deleting duplicates.
We measure cross-seed instability by the Jaccard distance between canonicalized clause sets from two seeds, and replay fidelity by the same distance between an original run and its deterministic replay.
Lower is better, and exact replay gives \(\Delta_{\Inv}^{\mathsf{replay}}(s)=0\).
The exact formulas are given in Appendix~\ref{app:replay-metrics}.

\subsection{Operational summary}
\label{subsec:capdr-alg}
In base mode, \tool follows the standard \pdr/\icthree loop and changes only proposal order. After the length-\(0\) check, it searches for a frontier CTI and creates a proof obligation. CP2 schedules pending obligations subject to fairness. Predecessor search either extends the counterexample chain or reaches an UNSAT query, after which CP1 orders candidate blockers and only guard-valid clauses are inserted. When the frontier closes, CP3 orders guarded push attempts. If \(F_k\equiv F_{k+1}\), the common frame is emitted as a candidate invariant. If enabled, post-fixpoint minimization retains only edits preserving the three SAFE-checker conditions. SAFE or UNSAFE is reported only after checker acceptance, and rejection is treated as \textsc{fail}. Full pseudocode is given in \cref{alg:capdr}.

\begin{algorithm}[h]
\caption{\small \tool core loop in base mode.}
\label{alg:capdr}
\scriptsize
\begin{algorithmic}[1]
\State \textbf{/* Trivial counterexample of length $0$ */}
\If{$\mathsf{SAT}(\Init(X)\wedge \neg\Prop(X))$}
  \State extract an initial bad state $x_0$ from the model
  \State set $\pi \gets (x_0)$
  \If{the UNSAFE checker accepts $\pi$}
    \State \Return UNSAFE with $\pi$
  \Else
    \State \Return \textsc{fail}
  \EndIf
\EndIf

\State Initialize $F_0 \gets \Init$; $F_1 \gets \Prop$; set $k \gets 0$
\State Create further frames lazily; when $F_{k+1}$ is first created, initialize it as $\Prop$
\State Initialize obligation multiset $\mathcal{Q} \gets \emptyset$

\While{true}
  \If{$\mathsf{SAT}(F_k(X)\wedge \Trans(X,U,X') \wedge \neg\Prop(X'))$}
    \State extract the bad-successor state cube $d$ over $X$ by unpriming the model's $X'$ assignment
    \State insert obligation $(k{+}1,d)$ into $\mathcal{Q}$
    \While{$\mathcal{Q}\neq\emptyset$}
      \State choose and remove an obligation $(i,d)$ from $\mathcal{Q}$ \Comment{policy-guided, fair}
      \If{$i = 0$ \textbf{or} $\mathsf{SAT}(\Init(X)\wedge d(X))$}
        \State reconstruct a trace $\pi$ ending in $\neg\Prop$ using the stored successor links and input witnesses
        \If{the UNSAFE checker accepts $\pi$}
          \State \Return UNSAFE with $\pi$
        \Else
          \State \Return \textsc{fail}
        \EndIf
      \EndIf

      \If{$\mathsf{SAT}(F_{i-1}(X)\wedge \Trans(X,U,X')\wedge d(X'))$}
        \State extract the predecessor state cube $d_{\mathsf{pred}}$ over $X$ from the model's $X$ assignment
        \State store the predecessor/input witness information
        \State reinsert $(i,d)$ into $\mathcal{Q}$ and insert $(i{-}1,d_{\mathsf{pred}})$ into $\mathcal{Q}$
      \Else
        \State construct candidate generalized blockers $\mathcal{C}(d,i)$,
        i.e., clauses $C=\neg d^\star$ for subcubes $d^\star\subseteq d$,
        with $\neg d\in\mathcal{C}(d,i)$
        \State order $\mathcal{C}(d,i)$ by policy score \Comment{fallback to deterministic order}
        \For{candidates $C$ in that order}
          \If{$\mathsf{UNSAT}(\Init(X)\wedge \neg C(X))$ \textbf{and}
          $\mathsf{UNSAT}(F_{i-1}(X)\wedge \Trans(X,U,X') \wedge \neg C(X'))$}
            \State add $C$ to frames $F_1,\dots,F_i$
            \State \textbf{break} \Comment{guaranteed by $C=\neg d$ after the initial-state check}
          \EndIf
        \EndFor
      \EndIf
    \EndWhile
  \Else
    \State attempt pushing learned clauses across frames $F_1,\dots,F_k$ using guard \cref{eq:push-check} \Comment{policy-guided order}
    \If{$F_k \equiv F_{k+1}$ ($\mathsf{canon}(F_k)=\mathsf{canon}(F_{k+1})$)}
      \State extract $\Inv \gets F_k$
      \State optionally minimize \(\Inv\) by keeping only edits that preserve the three SAFE-checker SAT conditions
      \If{the SAFE checker accepts $\Inv$}
        \State \Return SAFE with $\Inv$
      \Else
        \State \Return \textsc{fail}
      \EndIf
    \Else
      \State increment $k \gets k{+}1$ and initialize $F_{k+1}\gets \Prop$ if needed
    \EndIf
  \EndIf
\EndWhile
\end{algorithmic}
\end{algorithm}

It is worth noting that according to \cref{prop:base-blocker}, the blocking loop
in \cref{alg:capdr} is guaranteed to accept some guard-validated
candidate independently of the policy's ranking.

\section{Soundness Argument}
\label{sec:soundness}
What follows makes the trust boundary from \cref{subsec:separation} precise.
We first fix the checker model and the assumptions on the trusted components.
We then prove that checker acceptance implies semantic correctness of SAFE and UNSAFE certificates, and finally lift these certificate-level guarantees to \tool\ itself.
Consequently, the learned policy may affect search, but it does not participate in the correctness argument.

\vspace{-0.2cm}
\subsection{Checker model and assumptions}
\label{subsec:checker-defs}
This subsection fixes the checker-level objects used in the rest of \cref{sec:soundness}.
The following assumptions~A1-A2 are the only metatheoretic assumptions about the trusted checker from \cref{subsec:separation}: they state that its low-level reasoning and evaluation are correct.\smallskip

\noindent
\emph{Assumption A1 (Checker-query correctness):}
Answers returned by the SAT engine used inside the checker are correct: if UNSAT, the queried formula is unsatisfiable; if SAT, any returned assignment satisfies the queried formula.\smallskip

\noindent
\emph{Assumption A2 (Parsing and evaluation correctness):}
The checker faithfully parses and semantically encodes $(\Init,\Trans,\Prop)$ and the certificate, treats primed and unprimed variables consistently, correctly evaluates trace steps for UNSAFE certificates, and correctly implements any validation routine used for extension-mode definitional extensions.

Under assumptions~A1-A2, we define the acceptance conditions of the checkers used throughout the proofs.

\noindent
\emph{SAFE checker (base mode):}
Given a transition system $(\Init(X),\Trans(X,U,X'),\Prop(X))$ and a candidate invariant $\Inv(X)$,
the SAFE checker accepts if and only if all of the following are UNSAT:
\begin{align}
  &\Init(X) \wedge \neg\Inv(X) \label{eq:chk-init} \\
  &\Inv(X) \wedge \Trans(X,U,X') \wedge \neg\Inv(X') \label{eq:chk-step} \\
  &\Inv(X) \wedge \neg\Prop(X). \label{eq:chk-prop}
\end{align}
These are exactly the negations of~\cref{eq:init}-\cref{eq:prop}.

\noindent
\emph{UNSAFE checker:}
Given $(\Init,\Trans,\Prop)$ and a trace
$\pi=(x_0,u_0,x_1,\dots,u_{k-1},x_k)$,
the UNSAFE checker accepts if and only if:
\begin{itemize}
  \item $\Init(x_0)$ evaluates to true,
  \item for all $i<k$, $\Trans(x_i,u_i,x_{i+1})$ evaluates to true, and
  \item $\neg\Prop(x_k)$ evaluates to true.
\end{itemize}
Equivalently, it checks that $\pi$ is a valid finite execution ending in a bad state.

\noindent
\emph{SAFE checker (extension mode):}
When auxiliaries $Y$ are present, the emitted SAFE certificate is a pair $(\Def,\Inv(X,Y))$.
The extension-mode checker accepts if and only if
(i) it validates $\Def$ by a procedure whose correctness guarantees \cref{eq:defext-init,eq:defext-step}, and
(ii) it accepts $\Inv$ on the extended system $(\Init^{\mathsf{ext}},\Trans^{\mathsf{ext}},\Prop)$ using the SAFE checks~\cref{eq:chk-init}-\cref{eq:chk-prop} over state variables $(X,Y)$.

The next certificate-soundness results are semantic consequences of these
acceptance conditions. Assumptions~A1-A2 are needed only when relating those
mathematical conditions to the behavior of an actual checker implementation in
\cref{thm:capdr-sound}.

\vspace{-0.2cm}

\subsection{SAFE and UNSAFE certificates}
\label{subsec:theorem-safe}

\begin{restatable}[SAFE certificate soundness in base mode]{thm}{SafeCertSound}
\label{thm:safe-cert-sound}
If the SAFE checker in base mode accepts an invariant $\Inv$, then the transition system satisfies the safety property $\Prop$.
\end{restatable}

\begin{proof}
Assume the SAFE checker in base mode accepts $\Inv$.
Then, by definition of checker acceptance, the formulas
\cref{eq:chk-init,eq:chk-step,eq:chk-prop} are UNSAT.
Equivalently,
\begin{flalign}
 & \Init(X) \Rightarrow \Inv(X), \nonumber &\\
 & \Inv(X)\wedge \Trans(X,U,X') \Rightarrow \Inv(X') , \nonumber &\\
&  \Inv(X) \Rightarrow \Prop(X). \nonumber &
\end{flalign}

Let $x$ be any reachable state.
By definition of reachability, there exist $k\ge 0$, states
$x_0,\dots,x_k$, and inputs $u_0,\dots,u_{k-1}$ such that
$\Init(x_0)$, $\Trans(x_i,u_i,x_{i+1})$ for all $i<k$, and $x_k=x$.
We prove by induction on $i\in\{0,\dots,k\}$ that $\Inv(x_i)$ holds.

\emph{Base case ($i=0$):}
Since $\Init(x_0)$ and $\Init \Rightarrow \Inv$, we have $\Inv(x_0)$.

\emph{Inductive step:}
Assume $\Inv(x_i)$ holds for some $i<k$.
Since $\Trans(x_i,u_i,x_{i+1})$ holds and
$\Inv(X)\wedge \Trans(X,U,X') \Rightarrow \Inv(X')$, it follows that
$\Inv(x_{i+1})$ holds.

Thus $\Inv(x_k)$ holds.
Since $x_k=x$ and $\Inv \Rightarrow \Prop$, we obtain $\Prop(x)$.
Because $x$ was an arbitrary reachable state, every reachable state satisfies
$\Prop$.
Therefore the system is safe.
\end{proof}

\begin{restatable}[Soundness under conservative extension]{prp}{SoundnessConsExt}
\label{prop:defext-sound}
Assume $\Def$ satisfies~\cref{eq:defext-init}-\cref{eq:defext-step}, and let
$(\Init^{\mathsf{ext}},\Trans^{\mathsf{ext}},\Prop)$ be the extended system from \cref{subsec:cert-semantics}.
If the SAFE checker accepts an invariant $\Inv(X,Y)$ for the extended system, then the original transition system $(\Init(X),\Trans(X,U,X'),\Prop(X))$ is safe.
\end{restatable}

\begin{proof}
Let $x=x_k$ be reachable in the original system. Then there exist $k\ge 0$ and a sequence
$x_0,u_0,x_1,\dots,u_{k-1},x_k$
such that $\Init(x_0)$ and $\Trans(x_i,u_i,x_{i+1})$ for all $i<k$.

We prove by induction on the prefix length $j\in\{0,\dots,k\}$ that there exist
$y_0,\dots,y_j$ such that
$\Init^{\mathsf{ext}}(x_0,y_0)$ holds and
$\Trans^{\mathsf{ext}}(x_i,u_i,y_i,x_{i+1},y_{i+1})$ holds for all $i<j$.

For $j=0$, $\Init(x_0)$ and~\cref{eq:defext-init} yield some $y_0$ with $\Def(x_0,y_0)$, hence $\Init^{\mathsf{ext}}(x_0,y_0)$.

For the inductive step, assume the claim holds for some $j<k$. In particular, we have
$\Def(x_j,y_j)$ and $\Trans(x_j,u_j,x_{j+1})$.
By~\cref{eq:defext-step} there exists $y_{j+1}$ with $\Def(x_{j+1},y_{j+1})$.
By definition of $\Trans^{\mathsf{ext}}$, this implies
\[
  \Trans^{\mathsf{ext}}(x_j,u_j,y_j,x_{j+1},y_{j+1}).
\]
Thus the claim holds for $j+1$.

Therefore there exist $y_0,\dots,y_k$ such that
$(x_0,y_0),\dots,(x_k,y_k)$ is an execution of the extended system.
Since the SAFE checker accepts $\Inv(X,Y)$ on the extended system, checker acceptance gives
\[
  \Init^{\mathsf{ext}}(X,Y)\Rightarrow \Inv(X,Y),
\]
\[
  \Inv(X,Y)\wedge \Trans^{\mathsf{ext}}(X,U,Y,X',Y')
  \Rightarrow \Inv(X',Y'),
\]
and
\[
  \Inv(X,Y)\Rightarrow \Prop(X).
\]
Hence, by induction along the lifted execution, $\Inv(x_k,y_k)$ holds, and therefore $\Prop(x_k)$ holds.

Since $x_k=x$ was an arbitrary reachable original state, all reachable original states satisfy $\Prop$.
Thus the original transition system is safe.
\end{proof}

\begin{restatable}[UNSAFE certificate soundness]{thm}{UNSAFECertSound}
\label{thm:unsafe-cert-sound}
If the UNSAFE checker accepts a trace $\pi$, the transition system violates $\Prop$.
\end{restatable}

\begin{proof}
Let $\pi=(x_0,u_0,x_1,\dots,u_{k-1},x_k)$ be a trace accepted by the UNSAFE checker.
By the checker's definition:
(i) $\Init(x_0)$ holds,
(ii) for all $i<k$, $\Trans(x_i,u_i,x_{i+1})$ holds, and
(iii) $\neg\Prop(x_k)$ holds.

Items (i) and (ii) together show that $x_0,\dots,x_k$ is a concrete execution of the transition system.
Item (iii) shows the final state violates the safety property.
Therefore there exists a reachable state violating $\Prop$, so the system is unsafe.
\end{proof}

\subsection{Soundness of \tool as a certificate-producing procedure}
\label{subsec:capdr-soundness}
The previous theorems establish correctness of certificates \emph{independent} of how they were found.
We now connect this to \tool. A \tool run outputs either SAFE (with a candidate invariant $\Inv$), or UNSAFE (with a candidate trace $\pi$).
In both cases, \tool reports the claim as final output only after the corresponding checker accepts the emitted artifact; otherwise the run is treated as \textsc{fail}, not as a SAFE/UNSAFE result. This reporting discipline is stated in \cref{subsec:capdr-alg} and made explicit in the full pseudocode of \cref{alg:capdr}.

\begin{restatable}[\tool\ soundness]{thm}{toolsound}
\label{thm:capdr-sound}
Assuming assumptions~A1-A2, 
\begin{itemize}
  \item If \tool returns SAFE with a certificate accepted by the appropriate SAFE checker (base mode or extension mode), then the original transition system satisfies $\Prop$.
  \item If \tool returns UNSAFE with a certificate accepted by the UNSAFE checker, then the original transition system violates $\Prop$.
\end{itemize}
\end{restatable}

\begin{proof}
If \tool returns SAFE in base mode, it does so only after the base SAFE checker accepts $\Inv$.
By \cref{thm:safe-cert-sound}, the original system satisfies $\Prop$.

If \tool returns SAFE in extension mode, it does so only after the extension-mode SAFE checker accepts $(\Def,\Inv)$.
By definition of that checker, $\Inv$ is accepted on the extended system and $\Def$ is validated against a criterion that guarantees~\cref{eq:defext-init,eq:defext-step}.
Hence \cref{prop:defext-sound} applies, and the original system satisfies $\Prop$.

If \tool returns UNSAFE, it does so only after the UNSAFE checker accepts $\pi$.
By \cref{thm:unsafe-cert-sound}, the original system violates $\Prop$.
\end{proof}

\begin{restatable}{corol}{PolicyIndependence}
\label{cor:policy-indep}
Under assumptions~A1-A2, \tool is sound regardless of the learned policy.
\end{restatable}

\begin{proof}
The learned policy is consulted only inside the untrusted search
procedure. Whatever ranking it returns, every clause insertion, push, or
optional extension is committed only through the verifier's guarded
control logic, and \tool reports SAFE or UNSAFE only after the
corresponding independent checker accepts the emitted artifact. Hence the
certificate-level hypotheses used in \cref{thm:capdr-sound} are
unchanged by the choice of policy. Therefore every SAFE/UNSAFE result
returned by \tool is correct regardless of the learned policy.
\end{proof}

\subsection{Remark: completeness and progress}
\label{subsec:completeness-remark}
Soundness is independent of completeness.
An adversarial policy can still delay progress by prioritizing unhelpful candidates.
To recover the usual completeness argument of finite-state \pdr, the implementation
should ensure that the underlying \emph{baseline} search remains guard-complete:
(i) blocking candidate sets are nonempty, with $\neg d$ always available when
~\cref{eq:pred-query} is UNSAT,
(ii) if the scorer is unavailable, the engine reverts to a deterministic baseline order, and
(iii) the scheduler is fair in the sense that every pending obligation and every
pending push opportunity is eventually served by that baseline order
(bounded starvation suffices).

\section{Experimental Evaluation}
\label{sec:experiments}
\tool improves the checker-accepted artifact, not merely the
verifier's time-to-answer. We separate five questions: whether online
guidance improves the certified workflow objective, whether the raw SAFE
invariants become smaller and faster to check, whether the online-search
advantage survives the same strong post-hoc minimizer applied to both methods,
whether replay removes residual nondeterminism, and which guarded \pdr choice
points contribute most. The \tool\ (learned) configuration measures the
online-search contribution. The \tool+\textsc{inv-min} configuration is
reported separately as lightweight safe post-fixpoint minimization, while the
PDR$+$\cite{IvriiGurfinkelBelov2014SmallInvariants} and \tool$+$\cite{IvriiGurfinkelBelov2014SmallInvariants} configurations use one shared MSIS implementation.

\vspace{-0.3cm}
\subsection{Benchmarks, tool configurations, and protocol}
\label{subsec:eval-setup}

\subsubsection{Benchmarks}

We evaluate on the \emph{bit-level safety benchmarks} from the Hardware Model Checking Competition (HWMCC), which provides bit-level instances in the AIGER format for and-inverter graphs together with a workflow based on independently checkable certificates.
We use two disjoint benchmark corpora. The development corpus
\(\mathcal{B}_{\mathrm{dev}}\) is the deduplicated bit-vector
single-property AIGER subset of the HWMCC'20 benchmark archive~\cite{HWMCC20},
with all HWMCC'24 benchmark names and canonical AIGER hashes removed
(its SAFE subset contains 190 instances), and is used only for offline fitting,
model selection, and ablation. This is a natural development source for HWMCC:
it is competition-selected and bit-blasted to AIGER, matching the bit-level
setting while remaining disjoint from the certificate-mandated HWMCC'24 test
corpus. The held-out evaluation corpus \(\mathcal{B}_{\mathrm{eval}}\) is the
full 319-instance HWMCC'24 bit-level safety archive. This archive is a natural
target for \tool because HWMCC'24 made independently checkable certificates
part of the bit-level safety workflow~\cite{FroleyksYuPreinerBiereHeljanko2025CertificatesHWMCC}.
No HWMCC'24 evaluation instance contributes training examples, tuning decisions,
or online adaptation.

\subsubsection{Compared configurations}
We distinguish configurations used for development ablations from the frozen configurations used for the held-out headline comparison. Development-stage configurations are:
\emph{PDR (baseline)}, \emph{\tool\ (random ranking)}, \emph{\tool\ (CP1 only)}, \emph{\tool\ (CP1+CP2)}, \emph{\tool\ (CP1-CP3 full)}, \emph{\tool\ (CP1-CP3 full) + inv-min}, and deterministic \emph{\tool\ (CP1-CP3 full) + replay} for replay-fidelity evaluation only. Evaluation-stage configurations are \emph{PDR (baseline)}, \emph{\tool\ (learned)}, and \emph{\tool\ (learned) + inv-min}. The paired post-hoc experiment additionally uses PDR$+$\cite{IvriiGurfinkelBelov2014SmallInvariants} and \tool$+$\cite{IvriiGurfinkelBelov2014SmallInvariants} on the common SAFE subset; these are not new solving configurations, but the same MSIS postprocessor applied to checker-accepted raw outputs. The baseline is our implementation of~\cite{Bradley2011IC3}. All configurations use the same preprocessing, incremental SAT backend and interface, resource limits, certificate emitter, and independent checker.
The learned configuration differs from the baseline only in proposal ordering. The \textsc{inv-min} variant additionally removes duplicate or subsumed clauses and greedily attempts clause and literal deletions, retaining only edits for which \cref{eq:chk-init,eq:chk-step,eq:chk-prop} remain UNSAT; it makes no minimality claim.

\subsubsection{Implementation, training, and evaluation protocol}
\label{subsec:eval-impl}
We implement \tool\ as a prototype extension of a standard incremental \pdr/\icthree engine for AIGER bit-level safety. SAFE results are validated by a standalone SAT-based checker on the discovered CNF invariant, and UNSAFE traces by independent simulation. All solved counts and medians include only checker-accepted results. A verifier claim rejected by the checker is counted as \textsc{fail}, not as solved. UNSAFE traces receive no special learning-based quality treatment. They are included to maintain a uniform trust boundary in which every reported result is independently checked.

These runs are controlled internal comparisons rather than HWMCC standings. We do not claim that our baseline is the strongest available PDR engine. The host, preprocessing pipeline, solver integration, certificate representation, and checker differ from competition tools. In base mode, the emitted artifact is already a CNF invariant whose semantics is exactly the three SAT checks in \cref{eq:chk-init,eq:chk-step,eq:chk-prop}. Therefore \(t_{\mathsf{chk}}\) measures that artifact directly. Certifaiger checks competition witness circuits~\cite{FroleyksYuPreinerBiereHeljanko2025CertificatesHWMCC}. Using it here would first require a CNF-to-witness-circuit conversion and would mix conversion engineering with the effect of online ranking. We leave a common witness-circuit workflow to future work.

For CP1--CP3, the prototype uses the feature map described in \cref{subsec:policy-interface}.
At CP1, an action means trying a blocker \(C=\neg d^\star\) first within \(\mathcal{C}(d,i)\).
At CP2, it means servicing one ready obligation next.
At CP3, it means attempting one eligible push first.
The model never proposes a clause outside the verifier-generated CP1 set and never bypasses a guard.

Training examples are within-event candidate pairs. A pair is retained only when both alternatives have an observed development continuation cost or a rollout-estimated cost from the same checkpoint under the matched protocol in \cref{subsec:policy-interface}. The first proposal is the only intentionally changed verifier choice. The alternative with lower scalarized objective \(J\) is preferred. Timeouts, resource exits, and checker rejections receive the development-only failure treatment from Appendix~\ref{app:local-objective}.
The prototype uses the linear scorer and the \(L_2\)-regularized pairwise logistic loss from \cref{subsec:policy-interface} and Appendix~\ref{app:local-objective}.
In ablations, events from disabled choice points are omitted and the remaining scorer is trained by the same procedure.
The model is fitted once on \(\mathcal{B}_{\mathrm{dev}}\), frozen, and reused unchanged across every instance of \(\mathcal{B}_{\mathrm{eval}}\).
No evaluation benchmark contributes training examples, model selection, tuning decisions, or online adaptation.
Regularization and rollout-budget choices are made only on \(\mathcal{B}_{\mathrm{dev}}\).

Each campaign records the split manifest, canonical hashes, random seeds, replay logs, frozen-model identifiers, checker outcomes, and per-instance logs used to regenerate \cref{tab:overall,tab:msis-paired,tab:ablation,fig:tradeoff,fig:checker-cdf,fig:stability}. Reported learned-policy runtimes include feature extraction, scoring, and all guard checks, but exclude the one-time offline fitting cost. For every configuration labeled \(+\textsc{inv-min}\), the reported
runtime \(t\) includes the post-fixpoint minimization pass. We ran our experiments on a single Windows~11 machine with an Intel Core i7-11800H CPU and 64\,GB RAM, with a per-instance timeout of 1 hour and a memory limit of 56\,GB.

\subsubsection{Metrics}
We report the vector metrics \((t,\mathrm{size}(\pi),t_{\mathsf{chk}})\) from \cref{subsec:metrics-objective} and, when useful, the scalar objective \(J=\alpha t+\beta\,\mathrm{size}(\pi)+\gamma t_{\mathsf{chk}}\) with \((\alpha,\beta,\gamma)=(1,10^{-3},1)\). The factor \(10^{-3}\) weights 1000 certificate-size units as one second of scalar cost. The weights are fixed before evaluation, and the underlying vector metrics remain primary.
For SAFE instances, \(\mathrm{size}(\pi)=\mathrm{lit}(\Inv)\) and \(t_{\mathsf{chk}}\) is measured by the standalone CNF-invariant checker. Since \(\mathrm{size}(\pi)\) is outcome-specific, mixed SAFE/UNSAFE summaries in \cref{tab:overall} are workflow-level summaries. SAFE-only certificate quality is reported in \cref{tab:msis-paired,tab:ablation,fig:tradeoff,fig:checker-cdf}.

\subsubsection{Stability protocol}
Unless stated otherwise, solve counts and the medians of \(t\), \(\mathrm{size}(\pi)\), and \(t_{\mathsf{chk}}\) are reported for the fixed primary evaluation seed. The 10-seed campaign is used only for instability and replay-fidelity measurements.
Each configuration is run under 10 seeds while varying only low-level solver randomness; verifier-level tie-breaking remains fixed except in the random-ranking control. For each SAFE benchmark, we aggregate all pairwise cross-seed distances $\Delta_{\Inv}^{\mathsf{seed}}(s,s')$ by the median. Replay is evaluated separately by deterministic re-execution from the recorded log and measuring $\Delta_{\Inv}^{\mathsf{replay}}$ against the original artifact.
Distances are computed only for seed pairs in which both runs produce SAFE
certificates for the same benchmark. Benchmarks with fewer than two SAFE
certificates for a configuration are omitted from the distance distribution,
while coverage is reported separately. \Cref{tab:ablation} reports the median of these per-benchmark values
over eligible benchmarks, whereas \cref{fig:stability} shows their
distribution.

\vspace{-0.3cm}
\subsection{Results}
\label{subsec:eval-main}

\subsubsection{Overall end-to-end results}
\label{subsec:eval-overall}

\begin{table}[t]
\centering
\caption{\footnotesize
Held-out results on the 319 HWMCC'24 bit-level safety benchmarks.
Medians are over each configuration's checker-accepted solved set. The
``Unsolved'' column includes timeouts, resource exits, and checker rejections. Lower is better for \(t\), \(\mathrm{size}(\pi)\), \(t_{\mathsf{chk}}\), and \(J\).
The reported median \(J\) is the median of per-run scalarized costs, not
\(J\) evaluated on the displayed component medians.}
\label{tab:overall}
\small
\setlength{\tabcolsep}{4.5pt}
\resizebox{\columnwidth}{!}{%
\begin{tabular}{lrrrrrr}
\toprule
\textbf{Configuration} &
\textbf{Solved (SAFE/UNSAFE)} &
\textbf{Unsolved} &
\textbf{Med.\ $t$ [s]} &
\textbf{Med.\ $\mathrm{size}(\pi)$} &
\textbf{Med.\ $t_{\mathsf{chk}}$ [s]} &
\textbf{Med.\ $J$} \\
\midrule
PDR (baseline)      & 287 & 32 & 109.7 & 19365 & 1.24 & 151.1 \\
\tool (learned)     & 293 & 26 & 115.2 & 14596 & 0.63 & 142.4 \\
\tool + inv-min     & 293 & 26 & 117.5 & 12785 & 0.55 & 140.5 \\
\bottomrule
\end{tabular}}
\end{table}

\cref{tab:overall} reports certified end-to-end performance on the held-out evaluation benchmarks under the disjoint development/evaluation protocol from \cref{subsec:eval-impl}.
The key online-search comparison is the baseline row against
\tool\ (learned): neither row uses post-fixpoint invariant minimization. Thus
the certificate reduction in this comparison comes from CP1--CP3 guidance during
\pdr search, not from a final cleanup pass.

Under the primary evaluation seed, \tool reduces unsolved instances from \(32\) to \(26\) (six additional checker-accepted solutions), while improving the deliverable substantially: the median certificate proxy shrinks by \(24.6\%\) (19{,}365 \(\rightarrow\) 14{,}596) and the median independent checker time drops by \(49\%\) (1.24\,s \(\rightarrow\) 0.63\,s, about \(2\times\) faster checking). The median solving time increases by about \(5\%\) (109.7\,s \(\rightarrow\) 115.2\,s), consistent with spending additional SAT effort on lemma generalization and pushing. Under the fixed
weights \((\alpha,\beta,\gamma)=(1,10^{-3},1)\), these certificate improvements yield a lower median objective
(151.1 \(\rightarrow\) 142.4).

Post-fixpoint minimization (\tool+\textsc{inv-min}) is complementary:
it further reduces the median certificate proxy to 12{,}785 ($34.0\%$ below baseline) and cuts median
checking time to 0.55\,s ($55.6\%$ below baseline), with only a small additional runtime overhead
(117.5\,s median) and unchanged coverage (293 solved).
Overall, \tool trades a small amount of solving time for substantially smaller and cheaper-to-check certificates.

Because the solved sets differ across configurations, the medians in \cref{tab:overall} should be read as workflow-level summaries over each configuration's solved instances. The paired SAFE-instance comparison in \cref{fig:tradeoff} controls for this by restricting to benchmarks solved by both baseline \pdr\ and \tool.

\subsubsection{Fair post-hoc MSIS comparison}
\label{subsec:eval-msis}
To isolate online search from post-fixpoint minimization, we apply one shared implementation of the combined BIG+NFN method defined by~\cite{IvriiGurfinkelBelov2014SmallInvariants} to the 180 SAFE benchmarks solved by both baseline \pdr and \tool. In the terminology of~\cite{IvriiGurfinkelBelov2014SmallInvariants}, the pipeline runs NEC, FEAS, a second NEC pass, and BigMSIS. It deletes clauses only and returns a subset-minimal safe inductive subset, without claiming minimum cardinality or literal count. Because each evaluated AIGER job has one bad output, the fixed safety clause $\neg\mathit{Bad}$, where $\mathit{Bad}$ denotes that output, is used as the initial necessary set. Every other fixed-point clause is a deletion candidate. Both inputs use the same SAT backend, canonical clause order, one-hour postprocessing limit, and independent checker. The original solving times are unchanged and are not repeated; the reported $t_{\mathsf{msis}}$ is the additional postprocessing time only. PDR$+$\cite{IvriiGurfinkelBelov2014SmallInvariants} completes on 178 of the 180 inputs and \tool$+$\cite{IvriiGurfinkelBelov2014SmallInvariants} on all 180. To keep the quality comparison paired, all medians in \cref{tab:msis-paired} are computed on the 178 instances completed by both. Every completed postprocessed invariant is accepted by the independent checker.

\begin{table}[h]
\centering
\caption{\footnotesize
Fair paired post-hoc minimization on the common HWMCC'24 SAFE subset. The $+$\cite{IvriiGurfinkelBelov2014SmallInvariants} rows use the same combined BIG+NFN implementation of~\cite{IvriiGurfinkelBelov2014SmallInvariants}. ``Completed'' is over all 180 paired inputs; all medians are over the 178 inputs completed by both. Raw lit. is measured before postprocessing; the remaining size and checker columns are measured after it. Lower is better.}
\label{tab:msis-paired}
\small
\setlength{\tabcolsep}{3.6pt}
\resizebox{\columnwidth}{!}{%
\begin{tabular}{lrrrrrr}
\toprule
\textbf{Configuration} &
\textbf{Completed} &
\textbf{Raw lit.} &
\textbf{MSIS clauses} &
\textbf{MSIS lit.} &
\textbf{$t_{\mathsf{msis}}$ [s]} &
\textbf{$t_{\mathsf{chk}}$ [s]} \\
\midrule
PDR$+$\cite{IvriiGurfinkelBelov2014SmallInvariants}   & 178/180 & 19{,}842 & 3{,}274 & 14{,}026 & 11.4 & 0.78 \\
\tool$+$\cite{IvriiGurfinkelBelov2014SmallInvariants} & 180/180 & 14{,}731 & 2{,}807 & 11{,}741 &  7.6 & 0.48 \\
\bottomrule
\end{tabular}}
\end{table}

From the displayed paired medians, the shared minimizer removes 29.3\% of the baseline's raw literals and 20.3\% of \tool's. The larger proportional reduction for PDR is consistent with greater removable redundancy in its raw fixed points. Nevertheless, equal postprocessing does not erase the online-search aability-protdvantage: relative to PDR$+$\cite{IvriiGurfinkelBelov2014SmallInvariants}, \tool$+$\cite{IvriiGurfinkelBelov2014SmallInvariants} has 16.3\% fewer literals and 14.3\% fewer clauses, with 38.5\% lower checker time and 33.3\% less postprocessing time. Thus the common minimizer narrows, but does not close, the gap. This suggests that \tool often reaches a different useful clause core, rather than merely producing an invariant with less removable redundancy.

\subsubsection{Ablation on choice points and stability}
\label{subsec:eval-ablation}

\begin{table}[h]
\centering
\caption{\footnotesize
Development-only SAFE-invariant ablation on \(\mathcal{B}_{\mathrm{dev}}\).
Medians are over each row's certified SAFE set. The complement of the certified count comprises timeouts, resource exits, and checker rejections, all treated as \textsc{fail}. The last column lists cross-seed Jaccard distance, except for the replay row where replay distance is given. Lower is better. For the replay row, the runtime, invariant-size, and checker-time
entries are those of the original full-policy runs; only the final
column reports replay fidelity.
}
\label{tab:ablation}
\small
\setlength{\tabcolsep}{4.5pt}
\resizebox{\columnwidth}{!}{%
\begin{tabular}{lrrrrr}
\toprule
\textbf{Configuration} &
\textbf{Certified SAFE} &
\textbf{Med.\ $t$ [s]} &
\textbf{Med.\ $\mathrm{lit}(\Inv)$} &
\textbf{Med.\ $t_{\mathsf{chk}}$ [s]} &
\textbf{Med.\ $\Delta_{\Inv}$} \\
\midrule
PDR (baseline)                    & 176/190 & 123.2 & 16154 & 3.51 & 0.43 \\
\tool (random ranking)            & 171/190 & 135.4 & 17917 & 3.75 & 0.50 \\
\tool (CP1 only: blockers)         & 175/190 & 125.1 & 11314 & 2.21 & 0.29 \\
\tool (CP1+CP2: +obligation order)  & 178/190 & 131.4 & 10336 & 2.01 & 0.25 \\
\tool (CP1-CP3 full)               & 179/190 & 125.3 &  9620 & 1.61 & 0.20 \\
\tool (CP1-CP3 full) + inv-min     & 179/190 & 127.8 &  7990 & 1.31 & 0.20 \\
\tool (CP1-CP3 full) + replay      & 179/190 & 125.3 &  9620 & 1.61 & 0.00 \\
\bottomrule
\end{tabular}}
\end{table}

\begin{figure*}[t]
  \centering
  \subfloat[Per-instance tradeoff on SAFE benchmarks]{
    \includegraphics[height=0.175\linewidth]{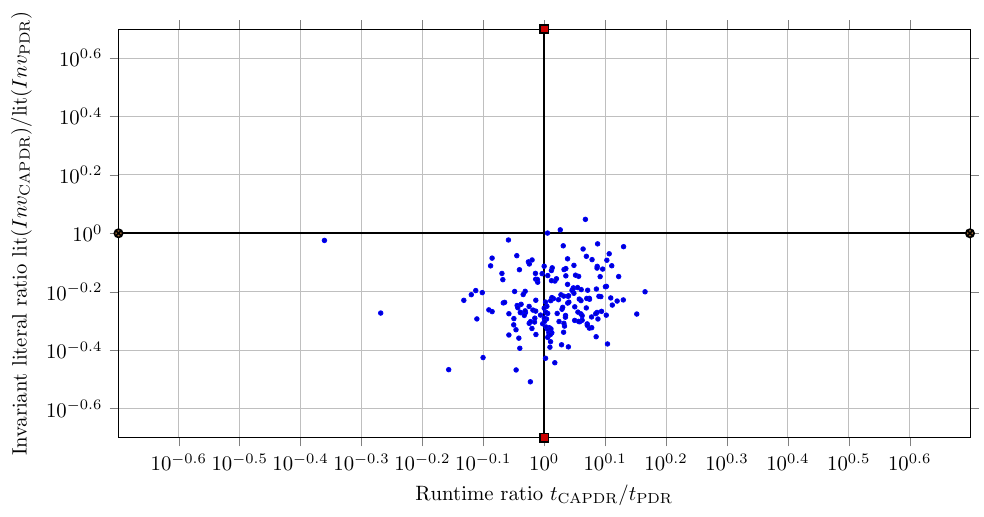}  
    \label{fig:tradeoff}}
  \hfil
  \subfloat[Checker time distribution on SAFE certificates]{
    \includegraphics[height=0.175\linewidth]{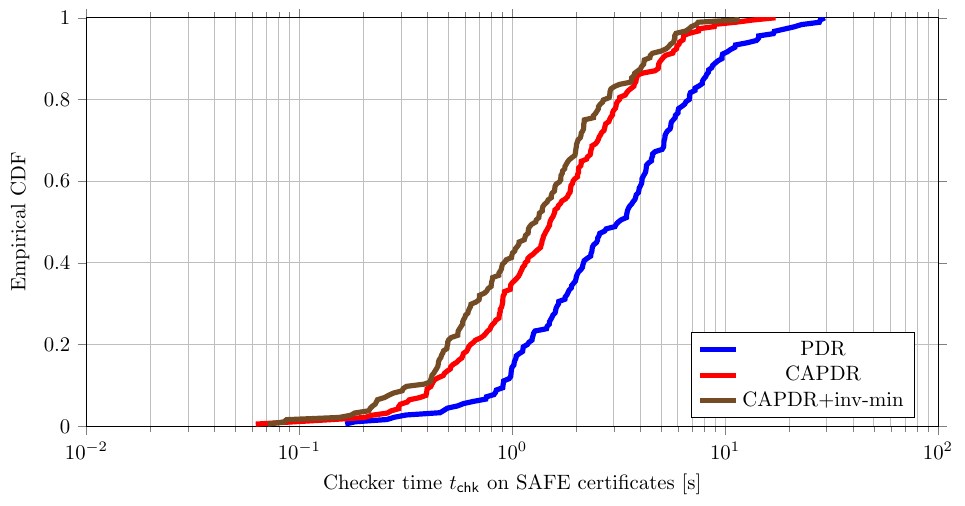}%
    \label{fig:checker-cdf}}
  \hfil
  \subfloat[Cross-seed instability and replay fidelity]{
    \includegraphics[height=0.175\linewidth]{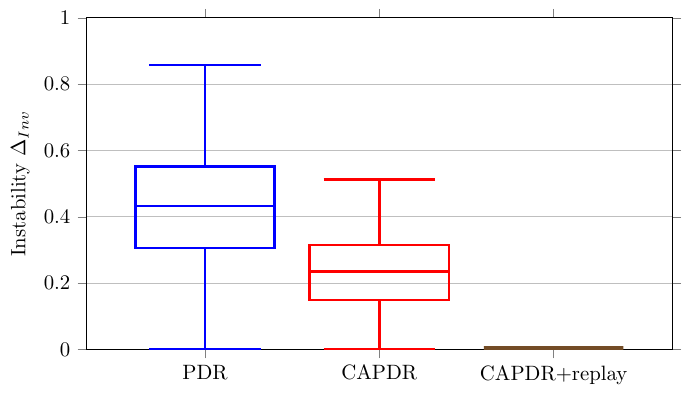}%
    \label{fig:stability}}   
\caption{\footnotesize
SAFE-certificate quality and reproducibility.
Panels (a) and (b) use the held-out HWMCC'24 evaluation corpus.
(a) Common 180-instance SAFE subset, with runtime and invariant-size
ratios relative to baseline; parity is \(1\).
(b) Empirical CDF of SAFE checker time on the same common subset,
where left is faster.
(c) On \(\mathcal{B}_{\mathrm{dev}}\), cross-seed instability for PDR
and \tool, and replay-fidelity distance for
\tool+\textsc{replay}.}
\end{figure*}

On the development split, \cref{tab:ablation} attributes the certificate gains on SAFE instances to specific \pdr choice points and serves both as an ablation study and as a characterization of the policy-design space used for offline fitting.
A random ranking policy is a useful negative control: it degrades solvability (171/190 vs.\ 176/190) and
inflates both $\mathrm{lit}(\Inv)$ and $t_{\mathsf{chk}}$, confirming that \tool's gains are not an artifact of
nondeterminism but of informed decisions. As in \cref{tab:overall}, the medians in \cref{tab:ablation} are computed over each configuration's own certified SAFE set and should therefore be interpreted jointly with the coverage column, rather than as paired common-subset comparisons.

Guiding CP1 (blocker selection/generalization) provides most of the benefit:
CP1 alone reduces the median invariant literal count by $30\%$ (16{,}154 $\rightarrow$ 11{,}314) and cuts
median checker time by $37\%$ (3.51\,s $\rightarrow$ 2.21\,s) with essentially unchanged runtime.
Adding guidance at CP2-CP3 yields further improvements, reaching a $40\%$ reduction in $\mathrm{lit}(\Inv)$
(16{,}154 $\rightarrow$ 9{,}620) and a $54\%$ reduction in checker time (3.51\,s $\rightarrow$ 1.61\,s) for the
full policy, while also solving three additional SAFE instances (179/190). Stability improves as well: the median cross-seed Jaccard distance drops from $0.43$ to $0.20$
(roughly $80\%$ clause-set overlap after canonicalization).
Replay addresses a different question: deterministic re-execution from the
recorded logs yields \(\mathrm{med}\,\Delta_{\Inv}^{\mathsf{replay}} = 0\) on
the same SAFE runs, confirming exact artifact-grade reproducibility.

The CP1-only row is not a dedicated shortest-first control, so \cref{tab:ablation} does not establish a learned-versus-shortest CP1 result. The additional gains from CP2 and CP3 show only that a rule defined solely by CP1 clause length is not a complete explanation of the full policy. A direct deterministic shortest-first ablation is not part of the present evaluation. The equality-controlled comparison with the dedicated postprocessor of~\cite{IvriiGurfinkelBelov2014SmallInvariants} is reported separately in \cref{subsec:eval-msis}.

\subsubsection{External coverage sanity check}
\label{subsec:legend-sanity}
On the 200-instance test list published with LeGend~\cite{miao2026legend}, a separate author run of \tool solved 189 instances under a 7200\,s per-instance timeout. LeGend reports 181 solved with IC3ref and 182 with ABC~\cite{miao2026legend}. This is a coverage-only sanity check, not a controlled tool comparison. Hosts, preprocessing, solver integration, lemma generation, certificate emission, and checker interfaces differ, and coverage is not the objective optimized here. The result therefore supports only the limited conclusion that the certificate-oriented evaluation is not accompanied by an obvious loss of broad solved coverage on that list.

\subsubsection{Per-instance tradeoffs, checker cost, and stability}
\label{subsec:eval-figs}
\cref{fig:tradeoff} provides a per-instance view on SAFE benchmarks solved by both baseline \pdr and \tool.
Of the 180 paired SAFE instances, 177 lie below \(y=1\) and three lie above it. The concentration below parity while runtime remains near \(x=1\) shows that the reductions are broad rather than driven by a few outliers. The three points above \(y=1\) are genuine checker-accepted cases, not checker errors. A locally valid clause order can steer \pdr toward a different proof trajectory and a larger final fixed point. With only three such cases, we cannot identify a reliable benchmark pattern. This non-monotonicity motivates the tunable objective in \cref{eq:objective-revised} rather than a claim that every local certificate-aware choice improves every instance.

\cref{fig:checker-cdf} shows the corresponding leftward shift
in independent checker time, with \tool+\textsc{inv-min} improving further.

\cref{fig:stability} quantifies how sensitive SAFE certificates are to nondeterminism by plotting the
distribution (over SAFE benchmarks) of the invariant-instability metric $\Delta_{\Inv}$ across $10$ seeds,
where $\Delta_{\Inv}$ is the Jaccard distance between canonicalized clause sets (\cref{subsec:policy-interface}).
The baseline \pdr exhibits substantial run-to-run variability (wide spread), indicating that different seeds often lead to meaningfully different resulting invariants.
In contrast, \tool shifts the distribution downward and tightens it, showing that certificate-aware guidance
tends to produce more consistent invariants across runs.
Finally, enforcing replay collapses the distribution to $\Delta_{\Inv}=0$, demonstrating that the replay log
fully determinizes the run and enables artifact-grade reproducibility.

\balance
\section{Conclusion}
\label{sec:conclusion}
\tool separates offline learning from certified cross-instance deployment and reuses one frozen ranker across held-out target instances without adding it to the trusted computing base. At deployment, the frozen model only orders verifier-generated proposals, while \pdr guards control clause insertion and pushing and the independent checker controls final acceptance. On held-out HWMCC'24 bit-level safety benchmarks, this guidance reduces certificate bulk and checker time with modest solving overhead. Its advantage also remains after applying the same MSIS postprocessor to paired baseline and \tool SAFE outputs, while replay logs make runs artifact-reproducible. Transfer remains an empirical performance question, whereas soundness remains checker-relative and policy-independent.

Future work includes exporting witness circuits for a common Certifaiger workflow and studying certificate-aware choices for SMT-based infinite-state systems~\cite{MebsoutTinelli2016ProofCertificates}, liveness certificates~\cite{GriggioRoveriTonetta2021CertifyingSAT}, and theorem-prover-based certification~\cite{SindoniEtAl2025TheoremProver}.

\medskip
\noindent
\textbf{Acknowledgements.}
The research of Arman Ferdowsi was funded in whole or in part by the Austrian Science Fund (FWF) 10.55776/ESP1705325.
The research of Laura Kovács was funded in whole or in part by the ERC Consolidator Grant ARTIST 101002685, the WWTF Grant ForSmart 10.47379/ICT22007, and the SBA Research COMET Center SBA-K1 NGC managed by the FFG.

\medskip
\noindent
\textbf{Artifact availability.}
The source code is publicly available at \url{https://github.com/armanferdowsi/CAPDR}.

\bibliographystyle{IEEEtran}
\bibliography{references}

\clearpage
\nobalance
\appendices

\section{Frame preservation by guarded clause insertion}
\label{app:frame-preservation}
Suppose \(C\) satisfies \cref{eq:block,eq:init-block,eq:relind} for an
obligation \((i,d)\), and let \(\widehat F_j=F_j\wedge C\) for
\(1\le j\le i\), while \(\widehat F_j=F_j\) otherwise. The boundary
\(F_0\Rightarrow F_1\) is preserved by \(\Init\Rightarrow C\). Boundaries
\(1\le j<i\) are strengthened on both sides, and later boundaries are either
strengthened only in the antecedent or unchanged; hence monotonicity is
preserved.

For any closed consecution index \(j<k\), if \(j+1\le i\), the consequent
frame is strengthened by \(C\). Old consecution gives \(F_{j+1}(X')\), and
monotonicity gives \(F_j\Rightarrow F_{i-1}\), so
\cref{eq:relind} gives \(C(X')\). If \(j+1>i\), the consequent is unchanged
and the antecedent is only strengthened or unchanged, so old consecution
suffices. Finally, \(F_j\Rightarrow\Prop\) is preserved because strengthening
frames cannot introduce new states.

\emph{Push preservation:}
If \(C\in F_i\) and \(F_i(X)\wedge\Trans(X,U,X')\Rightarrow C(X')\), then adding \(C\) to \(F_{i+1}\) preserves monotonicity because \(F_i\Rightarrow C\), and it preserves consecution at level \(i\) by the push guard. Other closed consecution obligations are unchanged or only strengthened in the antecedent.

\section{Auxiliary-Variable Extension Mode}
\label{app:aux-extension}
Some strengthening techniques (e.g., extended-resolution-style reasoning) introduce auxiliary Boolean state variables $Y$.
To keep certificates unambiguous while preserving soundness for the \emph{original} transition system over $X$, we support two conservative modes:
\begin{itemize}
  \item \emph{Projection mode:} emit an $X$-only candidate invariant $\Inv_X(X)$
  obtained from the internal auxiliary-rich representation by some elimination
  procedure; correctness is not assumed from that elimination step itself, but
  solely from subsequent acceptance by the base SAFE checker.
  \item \emph{Extension mode:} emit an invariant $\Inv(X,Y)$ together with an
  auxiliary-state relation $\Def(X,Y)$.
\end{itemize}

In extension mode, let $\Def'(X',Y')$ denote the primed copy of $\Def(X,Y)$, and define
\begin{align}
\Init^{\mathsf{ext}} &\triangleq \Init(X)\wedge \Def(X,Y), \notag\\
\Trans^{\mathsf{ext}} &\triangleq \Trans(X,U,X')\wedge \Def(X,Y)
\wedge \Def'(X',Y'). \notag
\end{align}
We require $\Def$ to be a conservative extension of the original dynamics under projection to $X$:
\begin{align}
\forall x.\quad \Init(x) &\Rightarrow \exists y.\;\Def(x,y),
\label{eq:defext-init_App}\\
\forall x,u,x',y.\;
\Trans(x,u,x')\wedge \Def(x,y)
&\Rightarrow \exists y'.\;\Def(x',y').
\label{eq:defext-step_App}
\end{align}

In extension mode, the emitted SAFE certificate is the pair $(\Def,\Inv(X,Y))$.
The independent checker must validate both components:
(i) it checks $\Inv$ on the extended system $(\Init^{\mathsf{ext}},\Trans^{\mathsf{ext}},\Prop)$, and
(ii) it validates that $\Def$ witnesses a conservative auxiliary-variable extension of the original system by a dedicated procedure whose correctness guarantees \cref{eq:defext-init,eq:defext-step}.
Only after both checks succeed may the checker accept the certificate.
Under these conditions, acceptance of an extended certificate implies safety of the original system.
We formalize this projection argument in \cref{prop:defext-sound}.

Projection mode emits an $X$-only candidate invariant and validates it with the
base SAFE checker. Extension mode uses the auxiliary-aware checker above.

\section{Deferred Formalization of the Local Objective}
\label{app:local-objective}
Fix a ranking event \(e=(r,\phi,A)\) and a continuation control law \(\mu\) for the remainder of the run.
The law \(\mu\) includes the frozen scorer, fixed tie-breaking, deterministic fallback behavior, and the fairness rule.
For CP1 and CP3, selecting \(a\) first means trying \(a\) first while leaving its SAT guard active.
If the guard fails, execution continues according to \(\mu\).
For CP2, selecting \(a\) first means servicing that obligation next.

The global objective in \cref{eq:objective-revised} induces the ideal local action value
\begin{align}
  Q_J(e,a) \triangleq
  -\mathbb{E}\bigl[J \mid
  & e,\ a\text{ is selected first}, \notag\\
  & \mu\text{ is used thereafter}\bigr],
  \qquad a\in A.
\end{align}
where the expectation is over the remaining solver and system nondeterminism.
Larger \(Q_J(e,a)\) means lower expected final cost.

Development does not need to compute this expectation exactly.
For selected alternatives, it obtains a sample or estimate \(\widehat J_e(a)\) from a measured continuation or a checkpointed rollout beginning at the same proof state.
Only alternatives from the same event and continuation protocol are compared.
The pairwise data set is
\[
  \mathcal{D}
  =\left\{
    (e,a^+,a^-)
    \ \middle|\
    \widehat J_e(a^+)<\widehat J_e(a^-)
  \right\}.
\]
A continuation that times out, exits because of a resource limit, or
produces a checker-rejected artifact is assigned
\(J_{\mathsf{fail}}=\langle\text{actual value}\rangle\), chosen above
the largest completed checker-accepted development cost.
This convention is used only for training labels.
Reported solved-set medians in \cref{sec:experiments} contain only checker-accepted runs.

For each \((e,a^+,a^-)\in\mathcal{D}\), let
\(\Delta x_e=x_e(a^+)-x_e(a^-)\).
For the linear scorer \(g_\theta(x)=\theta^\top x\), the evaluated prototype minimizes
\begin{equation}
  \mathcal{L}(\theta)
  =
  \sum_{(e,a^+,a^-)\in\mathcal{D}}
  \log\!\left(1+\exp\!\left[-\theta^\top\Delta x_e\right]\right)
  +\frac{\lambda_{\mathsf{reg}}}{2}\lVert\theta\rVert_2^2.
  \label{eq:pairwise-loss}
\end{equation}
The loss is small when the lower-cost alternative has the larger score.
Thus \(g_\theta\) is trained as an order-consistent surrogate for \(Q_J\), not as a proof oracle.

During deployment, \tool computes only \(g_\theta(x_e(a))\) and sorts the finite set \(A\).
It does not compute \(Q_J\), \(\widehat J_e\), or \(J\) online.
For CP1 and CP3, candidates are scanned in deterministic score order until a guard-valid action is committed or the list is exhausted.
For CP2, the highest-ranked obligation is selected subject to fairness.
CP1 cannot deadlock after the initial-state check because \(\neg d\in\mathcal{C}(d,i)\) satisfies \cref{eq:init-block,eq:relind} under the hypotheses of \cref{prop:base-blocker}.

\section{Exact Instability and Replay-Fidelity Formulas}
\label{app:replay-metrics}
For SAFE certificates, let \(\mathsf{canon}(\Inv)\) be the clause set
obtained by sorting literals within each clause, deleting repeated
literals within a clause, and deleting duplicate clauses.

For seeds $s,s'$, let
$A_s=\mathsf{canon}(\Inv_s)$ and
$A_{s'}=\mathsf{canon}(\Inv_{s'})$.
The \emph{cross-seed instability} is
\begin{equation}
\Delta_{\Inv}^{\mathsf{seed}}(s,s') \triangleq
\begin{cases}
0, & A_s\cup A_{s'}=\varnothing,\\[0.2em]
1-\dfrac{|A_s\cap A_{s'}|}{|A_s\cup A_{s'}|}, & \text{otherwise.}
\end{cases}
\end{equation}
This is the Jaccard distance on canonicalized clause sets and measures run-to-run certificate instability across seeds.

For a run executed with seed $s$, let $\Inv_s^{\mathsf{replay}}$ be the certificate produced by deterministic re-execution from that run's replay log, and set
$R_s=\mathsf{canon}(\Inv_s^{\mathsf{replay}})$.
The \emph{replay-fidelity distance} is
\begin{equation}
\Delta_{\Inv}^{\mathsf{replay}}(s) \triangleq
\begin{cases}
0, & A_s\cup R_s=\varnothing,\\[0.2em]
1-\dfrac{|A_s\cap R_s|}{|A_s\cup R_s|}, & \text{otherwise.}
\end{cases}
\end{equation}
Exact replay corresponds to $\Delta_{\Inv}^{\mathsf{replay}}(s)=0$.
We use $\Delta_{\Inv}^{\mathsf{seed}}$ for stability evaluation and
$\Delta_{\Inv}^{\mathsf{replay}}$ for replay validation. They measure different phenomena and are reported separately.

\end{document}